\documentclass[a4paper,UKenglish]{lipics-v2018}

\usepackage{tikz}
\usepackage[Algorithm]{algorithm}
\usepackage{dirtytalk}
\usepackage[noend]{algcompatible}

\newtheorem{claim}[theorem]{Claim}
\newtheorem{fact}[theorem]{Fact}

\newcommand{\s}{{\lvert \Sigma \rvert}}

\newcommand{\LCS}{\mathit{\textup{LCS}}}
\newcommand{\WLCS}{\mathit{\textup{WLCS}}}
\newcommand{\rtlm}{\mathit{\textup{RTLM}}}

\newcommand{\pat}{\mathit{\textup{pat}}}

\newcommand{\sk}{\mathit{\textup{sk}}}
\newcommand{\Oh}{\mathit{\mathcal{O}}}
\newcommand{\tOh}{\mathit{\mathcal{\widetilde O}}}

\nolinenumbers

\usepackage{microtype}

\title{Sketching, Streaming, and Fine-Grained Complexity of (Weighted) LCS}

\titlerunning{Sketching, Streaming, and Fine-Grained Complexity of (Weighted) LCS}

\author{Karl Bringmann}{Max Planck Institute for Informatics, Saarland Informatics Campus,\\ Saarbr\"ucken, Germany }{kbringma@mpi-inf.mpg.de}{}{}

\author{Bhaskar Ray Chaudhury}{Max Planck Institute for Informatics, Saarland Informatics Campus,\\ Graduate School of Computer Science, Saarbr\"ucken, Germany}{braycha@mpi-inf.mpg.de}{}{}

\authorrunning{K. Bringmann and B. R. Chaudhury}

\Copyright{Karl Bringmann and Bhaskar Ray Chaudhury}



\EventEditors{Sumit Ganguly and Paritosh Pandya}

\EventNoEds{2}

\EventLongTitle{38th IARCS Annual Conference on Foundations of Software Technology and Theoretical Computer Science (FSTTCS 2018)}

\EventShortTitle{FSTTCS 2018}

\EventAcronym{FSTTCS}

\EventYear{2018}

\EventDate{December 11--13, 2018}

\EventLocation{Ahmedabad, India}

\EventLogo{}

\SeriesVolume{122}

\ArticleNo{40}

\begin{document}

\maketitle

\begin{abstract}
  We study sketching and streaming algorithms for the Longest Common Subsequence problem (LCS) on strings of small alphabet size $|\Sigma|$. For the problem of deciding whether the LCS of strings $x,y$ has length at least $L$, we obtain  a sketch size and streaming space usage of $\Oh(L^{|\Sigma| - 1} \log L)$.
  We also prove matching unconditional lower bounds.
  
  As an application, we study a variant of LCS where each alphabet symbol is equipped with a weight that is given as input, and the task is to compute a common subsequence of maximum total weight. Using our sketching algorithm, we obtain an $\Oh(\min\{nm, n + m^\s\})$-time algorithm for this problem, on strings $x,y$ of length $n,m$, with $n \ge m$. We prove optimality of this running time up to lower order factors, assuming the Strong Exponential Time Hypothesis.
\end{abstract}

\section{Introduction}

\subsection{Sketching and Streaming LCS}

In the Longest Common Subsequence problem (LCS) we are given strings $x$ and $y$ and the task is to compute a longest string $z$ that is a subsequence of both $x$ and $y$.
This problem has been studied extensively, since it has numerous applications in bioinformatics (e.g.\ comparison of DNA sequences~\cite{AltschulGMML90}), natural language processing (e.g.\ spelling correction~\cite{Morgan70,WagnerF74}), file comparison (e.g.\ the UNIX \texttt{diff} utility~\cite{HuntMcI75,MillerM85}), etc. 
Motivated by big data applications, in the first part of this paper we consider space-restricted settings as follows:
\begin{itemize}
  \item \emph{LCS Sketching:} Alice is given $x$ and Bob is given $y$. Both also are given a number $L$. Alice and Bob compute sketches $\sk_L(x)$ and $\sk_L(y)$ and send them to a third person, the referee, who decides whether the LCS of $x$ and $y$ is at least $L$. The task is to minimize the size of the sketch (i.e., its number of bits) as well as the running time of Alice and Bob (encoding) and of the referee (decoding).
  
  \item \emph{LCS Streaming:} We are given $L$, and we scan the string $x$ from left to right once, and then the string $y$ from left to right once. After that, we need to decide whether the LCS of $x$ and $y$ is at least $L$. We want to minimize the space usage as well as running time. 
\end{itemize}
Analogous problem settings for the related edit distance have found surprisingly good solutions after a long line of work~\cite{belazzougui2016edit,jowhari2012efficient,saks2013space,chakraborty2016streaming}.
For LCS, however, strong unconditional lower bounds are known for sketching and streaming:
Even for $L = 4$ the sketch size and streaming memory must be $\Omega(n)$ bits, since the randomized communication complexity of this problem is $\Omega(n)$~\cite{sun2007communication}. 
Similarly strong results hold even for \emph{approximating} the LCS length~\cite{sun2007communication}, see also~\cite{liben2006finding}.
However, these impossibility results construct strings over alphabet size $\Theta(n)$. 

In contrast, in this paper we focus on strings $x,y$ defined over a fixed alphabet $\Sigma$ (of constant size). This is well motivated, e.g., for binary files ($\Sigma = \{0,1\}$), DNA sequences ($\Sigma = \{A,G,C,T\}$), or English text ($\Sigma = \{a,\ldots,z,A,\ldots,Z\}$ plus punctuation marks). We therefore \emph{suppress factors depending only on $|\Sigma|$} in $\Oh$-notation throughout the whole paper.
Surprisingly, this setting was ignored in the sketching and streaming literature so far; the only known upper bounds also work in the case of large alphabet and are thus $\Omega(n)$.

Before stating our first main result we define a run in a string as the non extendable repetition of a character. For example the string $baaabc$ has a run of character $a$ of length $3$. Our first main result is the following deterministic sketch.

\begin{theorem}
 \label{thm:sketchmain}
 Given a string $x$ of length $n$ over alphabet $\Sigma$ and an integer $L$, we can compute a subsequence $C_L(x)$ of $x$ such that (1) $|C_L(x)| = \Oh(L^\s)$, (2) $C_L(x)$ consists of $\Oh(L^{\s-1})$ runs of length at most $L$, and (3) any string $y$ of length at most $L$ is a subsequence of $x$ if and only if it is a subsequence of $C_L(x)$. Moreover, $C_L(x)$ is computed by a one-pass streaming algorithm with memory $\Oh(L^{\s-1} \log L)$ and running time $\Oh(1)$ per symbol of $x$.
\end{theorem}
Note that we can store $C_L(x)$ using $\Oh(L^{\s-1} \log L)$ bits, since each run can be encoded using $\Oh(\log L)$ bits. 
This directly yields a solution for \emph{LCS sketching}, where Alice and Bob compute the sketches $\sk_L(x) = C_L(x)$ and $\sk_L(y) = C_L(y)$ and the referee computes an LCS of $C_L(x)$ and $C_L(y)$. If this has length at least $L$ then also $x,y$ have LCS length at least $L$. Similarly, if $x,y$ have an LCS $z$ of length at least $L$, then $z$ is also a subsequence of $C_L(x)$ and $C_L(y)$, and thus their LCS length is at least $L$, showing correctness. The sketch size is $\Oh(L^{\s-1} \log L)$ bits, the encoding time is $\Oh(n)$, and the decoding time is $\Oh(L^{2\s})$, as LCS can be computed in quadratic time in the string length $\Oh(L^\s)$.

We similarly obtain an algorithm for \emph{LCS streaming} by computing $C_L(x)$ and then $C_L(y)$ and finally computing an LCS of $C_L(x)$ and $C_L(y)$.  The space usage of this streaming algorithm is $\Oh(L^{\s-1} \log L)$, and the running time is $\Oh(1)$ per symbol of $x$ and $y$, plus $\Oh(L^{2\s})$ for the last step.

These size, space, and time bounds are surprisingly good for $|\Sigma| = 2$, but quickly deteriorate with larger alphabet size. For very large alphabet size, this deterioration was to be expected due to the $\Omega(n)$ lower bound for $|\Sigma| = \Theta(n)$ from~\cite{sun2007communication}. 
We further show that this deterioration is necessary by proving optimality of our sketch in several senses:
\begin{itemize}
  \item We show that for any $L,\Sigma$ there exists a string $x$ (of length $\Oh(L^{\lvert \Sigma \rvert})$) such that no string~$x'$ of length $o(L^\s)$ has the same set of subsequences of length at most $L$. Similarly, this string $x$ cannot be replaced by any string consisting of $o(L^{\s-1})$ runs without affecting the set of subsequences of length at most $L$. 
This shows optimality of Theorem~\ref{thm:sketchmain} among sketches that replace $x$ by another string $x'$ (not necessarily a subsequence of $x$) and then compute an LCS of $x'$ and $y$. See Theorem~\ref{Combinatorial optimality of encoding}.

  \item More generally, we study the \emph{Subsequence Sketching} problem: Alice is given a string $x$ and number $L$ and computes $\sk_L(x)$. Bob is then given $\sk_L(x)$ and a string $y$ of length $L$ and decides whether $y$ is a subsequence of $x$. Observe that any solution for LCS sketching or streaming with size/memory $S = S(L,\Sigma)$ yields a solution for subsequence sketching with sketch size $S$.\footnote{For LCS sketching this argument only uses that we can check whether $y$ is a subsequence of $x$ by testing whether the LCS length of $x$ and $y$ is $|y|$. For LCS streaming we use the memory state right after reading $x$ as the sketch $\sk_L(x)$ and then use the same argument.} Hence, any lower bound for subsequence sketching yields a lower bound for LCS sketching and streaming. We show that any deterministic subsequence sketch has size $\Omega(L^{\s-1} \log L)$ in the worst case over all strings $x$. This matches the run-length encoding of $C_L(x)$ even up to the $\log L$-factor. If we restrict to strings of length $\Theta(L^{\s-1})$, we still obtain a sketch size lower bound of $\Omega(L^{\s-1})$. See Theorem~\ref{One sided deterministic communication}.
  
  \item Finally, randomization does not help either: We show that any randomized subsequence sketch, where Bob may err in deciding whether $y$ is a subsequence of $x$ with small constant probability, has size $\Omega(L^{\s-1})$, even restricted to strings $x$ of length $\Theta(L^{\s-1})$. See Theorem~\ref{One sided randomized communication}.
\end{itemize}
We remark that Theorem~\ref{thm:sketchmain} only makes sense if $L \ll n$. 
Although this is not the best motivated regime of LCS in practice, it corresponds to testing whether $x$ and $y$ are ``very different'' or ``not very different''. This setting naturally occurs, e.g., if one string is much longer than the other, since then $L \le m \ll n$. We therefore think that studying this regime is justified for the fundamental problem LCS.

\subsection{WLCS: In between min-quadratic and rectangular time}

As an application of our sketch, we determine the (classic, offline) time complexity of a weighted variant of LCS, which we discuss in the following.

A textbook dynamic programming algorithm computes the LCS of given strings $x,y$ of length $n$ in time $\Oh(n^2)$. 
A major result in fine-grained complexity shows that further improvements by polynomial factors would refute the Strong Exponential Time Hypothesis (SETH)~\cite{AbboudBVW15,BringmannK15} (see Section~\ref{sec:conlowerbound} for a definition). 
In case $x$ and $y$ have different lengths $n$ and $m$, with $n \ge m$, Hirschberg's algorithm computes their LCS in time $\Oh((n + m^2) \log n)$~\cite{Hirschberg77}, and this is again near-optimal under SETH.
This running time could be described as ``min-quadratic'', as it is quadratic in the minimum of the two string lengths.
In contrast, many other dynamic programming type problems have ``rectangular'' running time\footnote{By $\tOh$-notation we ignore factors of the form $\textup{polylog}(n)$.} $\tOh(nm)$, with a matching lower bound of $(nm)^{1-o(1)}$ under SETH, e.g., Fr\'echet distance~\cite{AltG95,Bringmann14}, dynamic time warping~\cite{AbboudBVW15,BringmannK15}, and regular expression pattern matching~\cite{myers1992four,backurs2016regular}. 

Part of this paper is motivated by the intriguing question whether there are problems with \emph{intermediate} running time, between ``min-quadratic'' and ``rectangular''. Natural candidates are generalizations of LCS, such as the weighted variant \emph{WLCS} as defined in~\cite{AbboudBVW15}: Here we have an additional \emph{weight function} $W \colon \Sigma \to \mathbb{N}$, and the task is to compute a common subsequence of $x$ and $y$ with maximum total weight. 
This problem is a natural variant of LCS that, e.g., came up in a SETH-hardness proof of LCS~\cite{AbboudBVW15}. It is not to be confused with other weighted variants of LCS that have been studied in the literature, such as a statistical distance measure where given the probability of every symbol's occurrence at every text location the task is to find a long and likely subsequence~\cite{amir2010weighted,cygan2016polynomial}, a variant of LCS that favors consecutive matches~\cite{lin2004rouge}, or edit distance with given operation costs~\cite{BringmannK15}.

Clearly, WLCS inherits the hardness of LCS and thus requires time $(n + m^2)^{1-o(1)}$. However, the matching upper bound $\tOh(n + m^2)$ given by Hirschberg's algorithm only works as long as the function $W$ is fixed (then the hidden constant depends on the largest weight). Here, we focus on the variant where the weight function $W$ is part of the input. In this case, the basic $\Oh(nm)$-time dynamic programming algorithm is the best known. 

Our second main result is to settle the time complexity of WLCS in terms of $n$ and $m$ for any fixed constant alphabet $\Sigma$, up to lower order factors $n^{o(1)}$ and assuming SETH.
\begin{theorem} \label{thm:main2}
  WLCS can be solved in time $\Oh(\min\{nm, n + m^{|\Sigma|}\})$. Assuming SETH, WLCS requires time $\min\{nm, n + m^{|\Sigma|}\}^{1-o(1)}$, even restricted to $n  =\Theta(m^\alpha)$ and $|\Sigma| = \sigma$ for any constants $\alpha \in \mathbb{R}, \alpha \ge 1$ and $\sigma \in \mathbb{N}, \sigma \ge 2$.
\end{theorem}

In particular, for $|\Sigma| > 2$ the time complexity of WLCS is indeed ``intermediate'',  in between ``min-quadratic'' and ``rectangular''! To the best of our knowledge, this is the first result of fine-grained complexity establishing such an intermediate running time.

To prove Theorem~\ref{thm:main2} we first observe that the usual $\Oh(nm)$ dynamic programming algorithm also works for WLCS. For the other term $n + m^{|\Sigma|}$, we compress $x$ by running the sketching algorithm from Theorem~\ref{thm:sketchmain} with $L = m$. This yields a string $x' = C_m(x)$ of length $\Oh(m^{|\Sigma|})$ such that WLCS has the same value on $(x,y)$ and $(x',y)$, since every subsequence of length at most $m$ of $x$ is also a subsequence of $x'$, and vice versa. Running the $\Oh(nm)$-time algorithm on $(x',y)$ would yield total time $\Oh(n + m^{|\Sigma| + 1})$, which is too slow by a factor $m$. To obtain an improved running time, we use the fact that $x'$ consists of $\Oh(m^{|\Sigma| - 1})$ runs. We design an algorithm for WLCS on a run-length encoded string $x'$ consisting of $r$ runs and an uncompressed string $y$ of length $m$ running time $\Oh(rm)$. This generalizes algorithms for LCS with one run-length encoded string~\cite{ann2008fast,freschi2004longest,liu2008finding}. Together, we obtain time $\Oh(\min\{nm, n + m^{|\Sigma|}\})$. 
We then show a matching SETH-based lower bound by combining our construction of incompressible strings from our sketching lower bounds (Theorem~\ref{Combinatorial optimality of encoding}) with the by-now classic SETH-hardness proof of LCS~\cite{AbboudBVW15,BringmannK15}. 

\subsection{Further Related Work}

Analyzing the running time in terms of multiple parameters like $n,m,L$ has a long history for LCS~\cite{Apostolico86,ApostolicoG87,EppsteinGGI92,Hirschberg77,HuntS77,IliopoulosR09,Myers86,NakatsuKY82,WuMMM90}. 
Recently tight SETH-based lower bounds have been shown for all these algorithms~\cite{BKsoda18}. In the second part of this paper, we perform a similar complexity analysis on a weighted variant of LCS. 
This follows the majority of recent work on LCS, which focused on transferring the early successes and techniques to more complicated problems, such as 
longest common increasing subsequence
\cite{moosa2013computing,kutz2011faster,yang2005fast,chan2007efficient},
tree LCS
\cite{mozes2009fast},
and many more generalizations and variants of LCS, see, e.g.,
\cite{kuboi2017,castelli2017,tiskin2006longest,jiang2004longest,
alber2002towards,landau2003sparse,keller2009longest,gotthilf2010restricted,
pevzner1992matrix,iliopoulos2007algorithms}.
For brevity, here we ignore the equally vast literature on the closely related edit distance.

\subsection{Notation}

For a string $x$ of length $n$ over alphabet $\Sigma$, we write $x[i]$ for its $i$-th symbol, $x[{i \ldots j}]$ for the substring from the $i$-th to $j$-th symbol, and $|x|$ for its length. For $c \in \Sigma$ we write $|x|_c := |\{i \mid x_i = c\}|$. For strings $x,y$ we write $x \circ y$ for their concatenation, and for $k \in \mathbb{N}$ we write $x^k$ for the $k$-fold repetition $x \circ \ldots \circ x$.
A subsequence of $x$ is any string of the form $y = x[{i_1}] \circ x[{i_2}] \circ \ldots \circ x[{i_\ell}]$ with $1 \le i_1 < i_2 < \ldots < i_\ell \le |x|$; in this case we write $y \preceq x$.
A \emph{run} in $x$ is a maximal substring $x[i \ldots j] = c^{j-i+1}$, consisting of a single alphabet letter $c \in \Sigma$.
Recall that we suppress factors depending only on $|\Sigma|$ in $\Oh$-notation.

\section{Sketching LCS}

In this section design a sketch for LCS, proving Theorem~\ref{thm:sketchmain}.
Consider any string $z$ defined over alphabet $S \subseteq \Sigma$. We call $z$ a \emph{$(q,S)$-permutation string} if we can partition $z = z^{(1)} \circ z^{(2)} \circ \ldots \circ z^{(q)}$ such that each $z^{(i)}$ contains each symbol in $S$ at least once. Observe that a $(q,S)$ permutation string contains any string $y$ of length at most $q$ over the alphabet $S$ as a subsequence. 

\begin{claim} \label{cla:coresketchlcs}
Consider any string $x = x' \circ c \circ x''$, where $x', x''$ are strings over alphabet $\Sigma$ and $c \in \Sigma$. Let $S \subseteq \Sigma$. 
If some suffix of $x'$ is an $(L,S)$-permutation string and $c \in S$, then for all strings $y$ of length at most $L$ we have $y \preceq x$ if and only if $y \preceq x' \circ x''$. 
\end{claim} 
  
\begin{proof} 
 The ``if''-direction is immediate. To prove the ``only if'', consider any subsequence $y$ of $x$ of length $d \leq L$ and let $y = x[{i_1}] \circ x[{i_2}] \circ \ldots \circ x[{i_d}]$. Let $\ell$ and $r$ be the length of $x'$ and $x''$, respectively. If $i_k \neq \ell+1$ for all $1 \le k \le d$, then clearly $y \preceq x' \circ x''$. Thus, assume that $i_k = \ell+1$ for some $k$. Let $a$ be minimal such that $x[{a \ldots \ell}]$ only contains symbols in~$S$. By assumption, $x[{a \ldots \ell}]$ is an $(L,S)$-permutation string, and $c = x[{\ell+1}] \in S$. Let $j \ge 1$ be the minimum index such that $x[{i_j}] \ldots x[{i_k}]$ only contains symbols in $S$. Since $j$ is minimal, $x[{i_{j-1}}] \notin S$ and thus $i_{b} < a$ for all $b<j$. Therefore $x[i_1] \circ x[i_2] \circ \ldots \circ x[i_{j-1}] \preceq x[0 \ldots a-1]$.  Since $x[{a \ldots \ell}]$ is an $(L,S)$-permutation string and $\lvert x[{i_j}] \circ \ldots \circ x[{i_k}] \rvert \leq d \le L$, it follows that $x[{i_j}] \circ \ldots \circ x[{i_k}]$ is a subsequence of $x[{a \ldots \ell}]$. Hence, $x[{i_1}] \circ \ldots \circ x[{i_k}] \preceq x'$ and $x[{i_{k+1}}] \circ \ldots \circ x[{i_d}] \preceq x''$, and thus $y \preceq x' \circ x''$. 
\end{proof} 
The above claim immediately gives rise to the following one-pass streaming algorithm.
\begin{algorithm}
\label{RoughSketch}
\begin{algorithmic}[1]
\STATE \textbf{initialize} $C_L(x)$ as the empty string
\FORALL{$i$ from $1$ to $\lvert x \rvert$}
  \IF{for all $S \subseteq \Sigma$ with $x[i] \in S$, no suffix of $C_L(x)$ is an $(L,S)$-permutation string}
    \STATE \textbf{set} $C_L(x)\gets C_L(x) \circ x[i]$
   \ENDIF
\ENDFOR
\STATE \textbf{return} $C_L(x)$
\end{algorithmic}
\caption{Outline for computing $C_L(x)$ given a string $x$ and an integer $L$}
\end{algorithm} 

By Claim~\ref{cla:coresketchlcs}, the string $C_L(x)$ returned by this algorithm satisfies the subsequence property (3) of Theorem~\ref{thm:sketchmain}. 
Note that any run in $C_L(x)$ has length at most $L$, since otherwise for $S = \{c\}$ we would obtain an $(L,S)$-permutation string followed by another symbol $c$, so that Claim~\ref{cla:coresketchlcs} would apply.
We now show the upper bounds on the length and the number of runs. Consider a substring $z = C_L(x)[{i \ldots j}]$ of $C_L(x)$, containing symbols only from $S \subseteq \Sigma$. We claim that $z$ consists of at most $ r_L(\lvert S \rvert ) := 2(L+1)^{\lvert S \rvert -1} -1$ runs. We prove our claim by induction on $\lvert S \rvert$. For $\lvert S \rvert = 1$, the claim holds trivially.  For $\lvert S\rvert > 1$ and any $k \ge 1$, let $i_k$ be the minimal index such that $z[{1 \ldots i_k}]$ is a $(k,S)$-permutation string, or $i_k = \infty$ if no such prefix of $z$ exists. Note that $i_L \ge |z|$, since otherwise a proper prefix of $z$ would be an $(L,S)$-permutation string, in which case we would have deleted the last symbol of $z$. 
The string $z[{i_{k-1}+1 \ldots i_k-1}]$ contains symbols only from $S \setminus \{z[{i_k}]\}$ and thus by induction hypothesis consists of at most $r_L(\lvert S \rvert -1)$ runs. 
Since $i_L \ge |z|$, we conclude that the number of runs in  $z$ is at most  
$L \cdot (r_L(\lvert S \rvert -1) + 1) \leq L \cdot 2 (L+1)^{|S|-2} \le 2 (L+1)^{\lvert S \rvert -1} -1 = r_L(|S|)$.
Thus the number of runs of $C_L(x)$ is at most $r_L(\s) \in \mathcal{O}(L^{\s -1})$, and since each run has length at most $L$ we obtain $\lvert C_L(x) \rvert \in \mathcal{O}(L^{\s})$. 

Algorithm 2 shows how to efficiently implement Algorithm 1 in time $\Oh(1)$ per symbol of~$x$. We maintain a counter $t_S$ (initialized to 0) and a set $Q_S$ (initialized to $\emptyset$) for every $S \subseteq \Sigma$ with the following meaning. After reading $x[{1 \ldots i}]$, let $j$ be minimal such that $x[{j \ldots i}]$ consists of symbols in $S$. Then $t_S$ is the maximum number $t$ such that $x[{j \ldots i}]$ is a $(t,S)$-permutation string. Moreover, let $k$ be minimal such that $x[{j\ldots k}]$ still is a $(t_S,S)$-permutation string. Then $Q_S \subseteq S$ is the set of symbols that appear in $x[{k+1 \ldots i}]$. In other words, in the future we only need to read the symbols in $S \setminus Q_S$ to complete a $(t_S + 1,S)$-permutation string. 
In particular, when reading the next symbol $x[{i+1}]$, in order to check whether Claim~\ref{cla:coresketchlcs} applies we only need to test whether for any $S \subseteq \Sigma$ with $x[i+1] \in S$ we have $t_S \ge L$. Updating $t_S$ and $Q_S$ is straightforward, and shown in Algorithm 2.

\begin{algorithm}[H]
\begin{algorithmic}[1]
  \STATE \textbf{set} $t_s \gets 0$, $Q_S \gets \emptyset$ for all $S \subseteq \Sigma$
  \STATE \textbf{set} $C_L(x)$ to the empty string
  \FORALL{$i$ from $1$ to $\lvert x \rvert$}
    \IF{$t_S < L$ for all $S \subseteq \Sigma$ with $x[i] \in S$}
    \STATE \textbf{set} $C_L(x) \gets C_L(x) \circ x[i]$
     \FORALL{$S$ such that $x[i] \in S$}
        \STATE \textbf{set} $Q_S \gets Q_S \cup \left\{x[i]\right\}$
        \IF{$Q_S = S$}
         \STATE \textbf{set} $Q_S \gets  \emptyset $
         \STATE \textbf{set} $t_S \gets t_S +1$
       \ENDIF
     \ENDFOR
     \FORALL{$S$ such that $x[i] \notin S$}
        \STATE \textbf{set} $t_S \gets 0$
        \STATE \textbf{set} $Q_S \gets \emptyset$
     \ENDFOR
    \ENDIF    
  \ENDFOR
\end{algorithmic}
\caption{Computing $C_L(x)$ in time $\Oh(1)$ per symbol of $x$}
\end{algorithm}

Since we assume $\s$ to be constant, each iteration of the loop runs in time $\Oh(1)$, and thus the algorithm determines $C_L(x)$ in time $\mathcal{O}(n)$.
This finishes the proof of Theorem~\ref{thm:sketchmain}.

\section{Optimality of the Sketch}
In this section we show that the sketch $C_L(x)$ is optimal in many ways. 
First, we show that the length and the number of runs are optimal for any sketch that replaces $x$ by any other string $z$ with the same set of subsequences of length at most $L$.
\begin{theorem}
\label{Combinatorial optimality of encoding}
For any $L$ and $\Sigma$ there exists a string $x$ such that for any string $z$ with $\{y \mid y \preceq x, \, |y| \le L\} = \{y \mid y \preceq z, \, |y| \le L\}$ we have $\lvert z \rvert = \Omega(L^{\s})$ and $z$ consists of $\Omega(L^{\s-1})$ runs. 
\end{theorem}

 Let $\Sigma = \left\{0,1, \ldots, \sigma-1 \right\}$ and $\Sigma_k = \left\{0,1,\ldots, k-1 \right\}$. We construct a family of strings $x^{(k)}$ recursively as follows, where $m := L / \s$:
\begin{align*} 
 &x^{(0)} = 0^{m} \\
 &x^{(k)} = (x^{(k-1)} \circ k)^m \circ x^{(k-1)} \quad \text{for } 1 \le k \le \sigma-1.
\end{align*} 

Theorem~\ref{Combinatorial optimality of encoding} now follows from the following inductive claim, for $k = \sigma-1$.
\begin{claim}
For any string $z$ with $\{y \mid y \preceq x^{(k)}, \, |y| \le m(k+1)\} = \{y \mid y \preceq z, \, |y| \le m(k+1)\}$ we have $\lvert z \rvert \geq m^{k+1}$ and the number of runs in $z$ is at least $m^{k}$.
\end{claim}
\begin{proof}
 We use induction on $k$. For $k=0$, since $y = 0^m \preceq x^{(0)}$ we have $z=0^{m'}$ with $m' \geq m$ and the number of runs in $z$ is exactly $1$. For any $k>0$, if $|x^{(k)}|_k > |z|_k$ then $k^{m} \preceq x^k$ but $ k^m \npreceq z$, and similarly if $|x^{(k)}|_k < |z|_k$ then $k^{m+1} \preceq z$ but $k^{m+1} \npreceq x^{(k)}$ (note that $m(k+1) \geq m+1$ since $k \geq 1$, and thus $y$ can be $k^{m+1}$). This implies $|z|_k = m$ and thus we have $z = z^{(0)} \circ k \circ z^{(1)} \circ k \circ \ldots  \circ k \circ z^{(m)}$, where each $z^{(i)}$ is a string on alphabet $\Sigma_{k-1}$. 
 Hence, for any $0 \le i \le m$ and string $y'$ of length at most $mk$, we have $y = k^i y' k^{m-i} \preceq z$ if and only if $y' \preceq z^{(i)}$. Similarly, $y \preceq x^{(k)}$ holds if and only if $y' \preceq x^{(k-1)}$. Since $y \preceq z$ is equivalent to $y \preceq x$ by assumption, we obtain that $y' \preceq z^{(i)}$ is equivalent to $y' \preceq x^{(k-1)}$. By induction hypothesis, $z^{(i)}$ has length at least $m^{k}$ and consists of at least $m^{k-1}$ runs. Summing over all~$i$, string $z$ has length at least $m^{k+1}$ and consists of at least $m^{k}$ runs.
\end{proof} 

Note that the run-length encoding of $C_L(x)$ has bit length $\mathcal{O}(L^{\s-1} \log L)$, since $C_L(x)$ consists of $\mathcal{O}(L^{\s-1})$ runs, each of which can be encoded using $\mathcal{O}(\log L)$ bits. We now show that this sketch has optimal size, even in the setting of \emph{Subsequence Sketching}: Alice is given a string $x$ of length $n$ over alphabet $\Sigma$ and a number $L$ and computes $\sk_L(x)$. Bob is then given $\sk_L(x)$ and a string $y$ of length at most\footnote{In the introduction, we used a slightly different definition where Bob is given a string of length \emph{exactly}~$L$. This might seem slightly weaker, but in fact the two formulations are equivalent (up to increasing $L$ by 1), as can be seen by replacing $x$ by $x' = 0^L 1 x$ and $y$ by $y' = 0^{L-|y|} 1 y$. Then $y \preceq x$ if and only if $y' \preceq x'$, and $y'$ has fixed length $L+1$.} $L$ and decides whether $y$ is a subsequence of~$x$.

We construct the following hard strings for this setting, similarly to the previous construction. 
Let $\Sigma = \left\{0,1,2, \ldots, \sigma-1 \right\}$ and $m \in \mathbb{N}$. Consider any vector $z \in \{0,\ldots,m-1\}^{k}$, where $k := m^{\sigma-1}$. We define the string $x = x(z)$ recursively as follows; see Figure \ref{x(z)string} for an illustration: 
\begin{align*}
x(z) &= x^{(\sigma-1,0)}\\
x^{(c,i)} &= \Big( \bigcirc_{j=0}^{m-2} x^{(c-1,m \cdot i + j)} \circ c \Big) \circ x^{(c-1,m \cdot i + m-1)} \quad \text{for }1 \leq c \leq \sigma-1 \\
x^{(0,i)} &= 0^{z[i]}  
\end{align*}
A straightforward induction shows that $| x(z) | \le m^{\sigma} - 1$. Moreover, for any $0 \le i < m^{\sigma-1}$ with base-$m$ representation $i = \sum_{j=0}^{\sigma-2} i_j \cdot m^j$, where $0 \le i_j < m$, we define the following string; see Figure \ref{patstring} for an illustration: 
\[ \pat(i,y) := \big(\bigcirc_{j=1}^{\sigma-1} (\sigma-j)^{i_{\sigma-1-j}}\big) \circ y \circ \big(\bigcirc_{j=1}^{\sigma-1} j^{m - 1-i_{j-1}}\big). \]

\begin{figure}
   \begin{tikzpicture}[scale=0.8]

\draw[black, very thick] (-8,0+2)rectangle (9,0.5+2);
\node at (0.5,0.25+2) {$\scriptstyle x(z)$};

\draw[black, very thick, fill=blue!10] (-8,0+1)rectangle (-3,0.5+1);
\node at (-5.5,0.25+1) {$\scriptstyle x^{(1,0)}$};

\draw[black, very thick, fill=green!10] (-3,0+1)rectangle (-2,0.5+1);
\node at (-2.5,0.25+1) {$\scriptstyle 2$};

\draw[black, very thick, fill=blue!10] (-2,0+1)rectangle (3,0.5+1);
\node at (0,0.25+1) {$\scriptstyle x^{(1,1)}$};

\draw[black, very thick, fill=green!10] (3,0+1)rectangle (4,0.5+1);
\node at (3.5,0.25+1) {$\scriptstyle 2$};

\draw[black, very thick, fill=blue!10] (4,0+1)rectangle (9,0.5+1);
\node at (6.5,0.25+1) {$\scriptstyle x^{(1,2)}$};

\draw[black, very thick, fill=red!10] (-8,0)rectangle (-7,0.5);
\node at (-7.5,0.25) {$\scriptstyle x^{(0,0)}$};

\draw[black, very thick, fill=blue!10] (-7,0)rectangle (-6,0.5);
\node at (-6.5,0.25) {$\scriptstyle 1$};

\draw[black, very thick, fill=red!10] (-6,0)rectangle (-5,0.5);
\node at (-5.5,0.25) {$\scriptstyle x^{(0,1)}$};

\draw[black, very thick, fill=blue!10] (-5,0)rectangle (-4,0.5);
\node at (-4.5,0.25) {$\scriptstyle 1$};

\draw[black, very thick, fill=red!10] (-4,0)rectangle (-3,0.5);
\node at (-3.5,0.25) {$\scriptstyle x^{(0,2)}$};

\draw[black, very thick, fill=green!10] (-3,0)rectangle (-2,0.5);
\node at (-2.5,0.25) {$\scriptstyle 2$};

\draw[black, very thick, fill=red!10] (-2,0)rectangle (-1,0.5);
\node at (-1.5,0.25) {$\scriptstyle x^{(0,3)}$};

\draw[black, very thick, fill = blue!10] (-1,0)rectangle (0,0.5);
\node at (-0.5,0.25) {$\scriptstyle 1$};

\draw[black, very thick, fill=red!10] (0,0)rectangle (1,0.5);
\node at (0.5,0.25) {$\scriptstyle x^{(0,4)}$};

\draw[black, very thick, fill=blue!10] (1,0)rectangle (2,0.5);
\node at (1.5,0.25) {$\scriptstyle 1$};

\draw[black, very thick, fill=red!10] (2,0)rectangle (3,0.5);
\node at (2.5,0.25) {$\scriptstyle x^{(0,5)}$};

\draw[black, very thick, fill=green!10] (3,0)rectangle (4,0.5);
\node at (3.5,0.25) {$\scriptstyle 2$};

\draw[black, very thick, fill=red!10] (4,0)rectangle (5,0.5);
\node at (4.5,0.25) {$\scriptstyle x^{(0,6)}$};

\draw[black, very thick, fill=blue!10] (5,0)rectangle (6,0.5);
\node at (5.5,0.25) {$\scriptstyle 1$};

\draw[black, very thick, fill=red!10] (6,0)rectangle (7,0.5);
\node at (6.5,0.25) {$\scriptstyle x^{(0,7)}$};

\draw[black, very thick, fill=blue!10] (7,0)rectangle (8,0.5);
\node at (7.5,0.25) {$\scriptstyle 1$};

\draw[black, very thick, fill=red!10] (8,0)rectangle (9,0.5);
\node at (8.5,0.25) {$\scriptstyle x^{(0,8)}$};

\draw[black, very thick, fill = red!10] (-8,0-1)rectangle (-7,0.5-1);
\node at (-7.5,0.25-1) {$\scriptstyle 0^{z[0]}$};

\draw[black, very thick, fill = blue!10] (-7,0-1)rectangle (-6,0.5-1);
\node at (-6.5,0.25-1) {$\scriptstyle 1$};

\draw[black, very thick, fill = red!10] (-6,0-1)rectangle (-5,0.5-1);
\node at (-5.5,0.25-1) {$\scriptstyle 0^{z[1]}$};

\draw[black, very thick, fill = blue!10] (-5,0-1)rectangle (-4,0.5-1);
\node at (-4.5,0.25-1) {$\scriptstyle 1$};

\draw[black, very thick, fill = red!10] (-4,0-1)rectangle (-3,0.5-1);
\node at (-3.5,0.25-1) {$\scriptstyle 0^{z[2]}$};

\draw[black, very thick, fill = green!10] (-3,0-1)rectangle (-2,0.5-1);
\node at (-2.5,0.25-1) {$\scriptstyle 2$};

\draw[black, very thick, fill = red!10] (-2,0-1)rectangle (-1,0.5-1);
\node at (-1.5,0.25-1) {$\scriptstyle 0^{z[3]}$};

\draw[black, very thick, fill = blue!10] (-1,0-1)rectangle (0,0.5-1);
\node at (-0.5,0.25-1) {$\scriptstyle 1$};

\draw[black, very thick, fill = red!10] (0,0-1)rectangle (1,0.5-1);
\node at (0.5,0.25-1) {$\scriptstyle 0^{z[4]}$};

\draw[black, very thick, fill = blue!10] (1,0-1)rectangle (2,0.5-1);
\node at (1.5,0.25-1) {$\scriptstyle 1$};

\draw[black, very thick, fill = red!10] (2,0-1)rectangle (3,0.5-1);
\node at (2.5,0.25-1) {$\scriptstyle 0^{z[5]}$};

\draw[black, very thick, fill = green!10] (3,0-1)rectangle (4,0.5-1);
\node at (3.5,0.25-1) {$\scriptstyle 2$};

\draw[black, very thick, fill = red!10] (4,0-1)rectangle (5,0.5-1);
\node at (4.5,0.25-1) {$\scriptstyle 0^{z[6]}$};

\draw[black, very thick, fill = blue!10] (5,0-1)rectangle (6,0.5-1);
\node at (5.5,0.25-1) {$\scriptstyle 1$};

\draw[black, very thick, fill = red!10] (6,0-1)rectangle (7,0.5-1);
\node at (6.5,0.25-1) {$\scriptstyle 0^{z[7]}$};

\draw[black, very thick, fill = blue!10] (7,0-1)rectangle (8,0.5-1);
\node at (7.5,0.25-1) {$\scriptstyle 1$};

\draw[black, very thick, fill= red!10] (8,0-1)rectangle (9,0.5-1);
\node at (8.5,0.25-1) {$\scriptstyle 0^{z[8]}$};

\end{tikzpicture} 
 
 \caption{Illustration of constructing $x(z)$ from $z$. Let $m = \sigma = 3$. Consider a string $z$ of length $m^{\sigma-1} = 9$. The figure shows the construction of $x(z)$ from $z$}
 \label{x(z)string}
\end{figure}
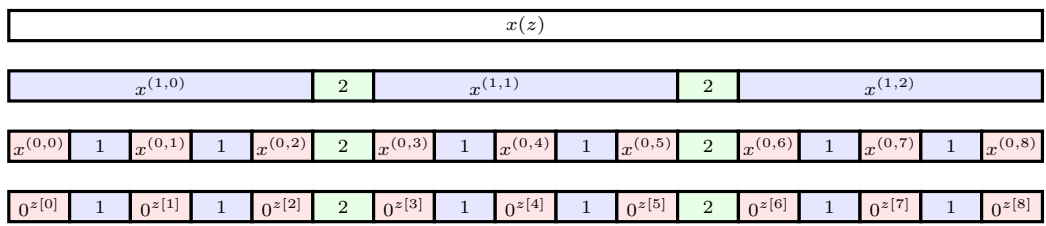

\begin{figure}
    \centering

\begin{tikzpicture}[scale=0.812]
\draw[black, very thick, fill=green!10] (0,-2)rectangle (1,-2.5);
\node at (0.5,-2.25) {$\scriptstyle 2$};

\draw[black, very thick, fill=blue!10] (1,-2)rectangle (2,-2.5);
\node at (1.5,-2.25) {$\scriptstyle 1$};

\draw[black, very thick, fill=red!10] (2,-2)rectangle (3,-2.5);
\node at (2.5,-2.25) {$\scriptstyle y$};

\draw[black, very thick, fill=blue!10] (3,-2)rectangle (4,-2.5);
\node at (3.5,-2.25) {$\scriptstyle 1$};

\draw[black, very thick, fill=green!10] (4,-2)rectangle (5,-2.5);
\node at (4.5,-2.25) {$\scriptstyle 2$};

\end{tikzpicture}

    \caption{Illustration of the construction of $pat(i,y)$. Let $m = \sigma =3$. Consider $i = 4 = 1 \cdot 3^{1} + 1 \cdot 3^{0}$. Therefore $\pat(i,y) = 21y12$.}
    \label{patstring}
\end{figure}
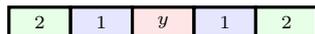

The following claim shows that testing whether $\pat(i,y)$ is a subsequence of $x(z)$ allows to infer the entries of $z$.

\begin{claim}
\label{pattern claim}
We have $\pat(i,y) \preceq x(z)$ if and only if $y \preceq 0^{z[i]}$.
\end{claim}

\begin{proof} 

\begin{figure}
    \centering

\begin{tikzpicture}[scale=0.8]

\draw[black, very thick, fill=red!10] (-8,0)rectangle (-7,0.5);
\node at (-7.5,0.25) {$\scriptstyle 0^{z[0]}$};

\draw[black, very thick, fill=blue!10] (-7,0)rectangle (-6,0.5);
\node at (-6.5,0.25) {$\scriptstyle 1$};

\draw[black, very thick, fill=red!10] (-6,0)rectangle (-5,0.5);
\node at (-5.5,0.25) {$\scriptstyle 0^{z[1]}$};

\draw[black, very thick, fill=blue!10] (-5,0)rectangle (-4,0.5);
\node at (-4.5,0.25) {$\scriptstyle 1$};

\draw[black, very thick, fill=red!10] (-4,0)rectangle (-3,0.5);
\node at (-3.5,0.25) {$\scriptstyle 0^{z[2]}$};

\draw[black, very thick, fill=green!10] (-3,0)rectangle (-2,0.5);
\node at (-2.5,0.25) {$\scriptstyle 2$};

\draw[black, very thick, fill=red!10] (-2,0)rectangle (-1,0.5);
\node at (-1.5,0.25) {$\scriptstyle 0^{z[3]}$};

\draw[black, very thick, fill=blue!10] (-1,0)rectangle (0,0.5);
\node at (-0.5,0.25) {$\scriptstyle 1$};

\draw[black, very thick, fill=red!10] (0,0)rectangle (1,0.5);
\node at (0.5,0.25) {$\scriptstyle 0^{z[4]}$};

\draw[black, very thick, fill=blue!10] (1,0)rectangle (2,0.5);
\node at (1.5,0.25) {$\scriptstyle 1$};

\draw[black, very thick, fill=red!10] (2,0)rectangle (3,0.5);
\node at (2.5,0.25) {$\scriptstyle 0^{z[5]}$};

\draw[black, very thick, fill=green!10] (3,0)rectangle (4,0.5);
\node at (3.5,0.25) {$\scriptstyle 2$};

\draw[black, very thick, fill=red!10] (4,0)rectangle (5,0.5);
\node at (4.5,0.25) {$\scriptstyle 0^{z[6]}$};

\draw[black, very thick, fill=blue!10] (5,0)rectangle (6,0.5);
\node at (5.5,0.25) {$\scriptstyle 1$};

\draw[black, very thick, fill=red!10] (6,0)rectangle (7,0.5);
\node at (6.5,0.25) {$\scriptstyle 0^{z[7]}$};

\draw[black, very thick, fill=blue!10] (7,0)rectangle (8,0.5);
\node at (7.5,0.25) {$\scriptstyle 1$};

\draw[black, very thick, fill=red!10] (8,0)rectangle (9,0.5);
\node at (8.5,0.25) {$\scriptstyle 0^{z[8]}$};

\draw[black, very thick, fill=green!10] (-2.5,-2)rectangle (-1.5,-2.5);
\node at (-2,-2.25) {$\scriptstyle 2$};

\draw[black, very thick, fill=blue!10] (-1.5,-2)rectangle (-0.5,-2.5);
\node at (-1,-2.25) {$\scriptstyle 1$};

\draw[black, very thick, fill=red!10] (-0.5,-2)rectangle (0.5,-2.5);
\node at (0,-2.25) {$\scriptstyle y$};

\draw[black, very thick, fill=blue!10] (0.5,-2)rectangle (1.5,-2.5);
\node at (1,-2.25) {$\scriptstyle 1$};

\draw[black, very thick, fill=green!10] (1.5,-2)rectangle (2.5,-2.5);
\node at (2,-2.25) {$\scriptstyle 2$};


\draw[green, ultra thick] (-2,-2) -- (-2.5,0);
\draw[green, ultra thick] (2,-2) -- (3.5,0);

\draw[blue, ultra thick] (-1,-2) -- (-0.5,0);
\draw[blue, ultra thick] (1,-2) -- (1.5,0);

\draw[red, ultra thick] (0,-2) -- (0.5,0);

\draw[black, very thick] (-2.5,-2) rectangle (2.5,-2.5);
\end{tikzpicture}
    \caption{Illustration of Claim~\ref{pattern claim}. Let $m = \sigma = 3$ and $i=4$. Then $\pat(i,y) = 21y12$. Now observe that $\pat(i,y) \preceq x(z)$ if and only if $y \preceq x^{(0,i)} = 0^{z[i]}$.}
    \label{matchingbetween patstring and x}
\end{figure}
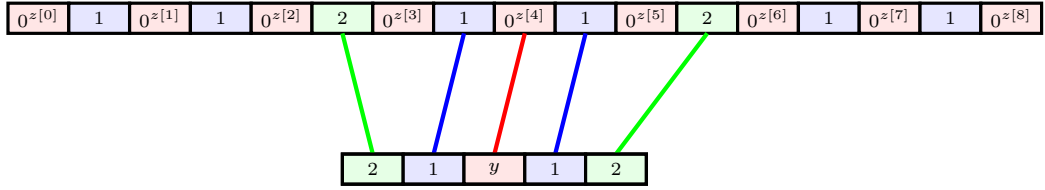

See Figure \ref{matchingbetween patstring and x} for illustration. Given $i$ and $y$, let $z^{(c)} = c^{i_{c-1}} \circ z^{(c-1)} \circ c^{m-1-i_{c-1}}$ for all $1 \leq c \leq \sigma-1$, and $z^{(0)} = y$. Note that $z^{(\sigma-1)}= \pat(i,y)$ . 
Set $j_c := \sum_{l=c}^{\sigma-2}i_l \cdot m^{l-c}$, so in particular we have $j_c = m \cdot j_{c-1} + i_c$.
Observe that $z^{(c)} \preceq x^{(c, j_c)}$ if and only if $z^{(c-1)} \preceq x^{(c-1, j_{c-1})}$, which follows immediately after matching all $c$'s in $z^{(c)}$ and $x^{(c,j_c)}$. Therefore, $\pat(i,y) = z^{(\sigma-1)} \preceq x^{(\sigma-1,0)} = x(z)$ holds if and only if $z^{(c)} \preceq x^{(c, j_c)}$ for any $c \leq \sigma-2$. Substituting $c = 0$ we obtain that $\pat(i,y) \preceq x(z)$ holds if and only if $y = z^{(0)} \preceq x^{(0, j_0)} = x^{(0,i)} = 0^{z[i]}$.
\end{proof}

\begin{theorem}
\label{One sided deterministic communication}
 Any deterministic subsequence sketch has size $\Omega(L^{\s -1} \log L)$ in the worst case.
 Restricted to strings of length $\Theta(L^{\s-1})$, the sketch size is $\Omega(L^{\s-1})$.
\end{theorem}

\begin{proof}
Let $m := L/|\Sigma|$.
Let $z \in \{0,\ldots,m-1\}^k$ with $k = m^{|\Sigma|-1}$ and let $x = x(z)$ as above. Alice is given $x,L$ as input. Notice that there are $m^k$ distinct inputs for Alice. Assume for contradiction that the sketch size is less that $k \cdot \log m$ for every $x$. Then the total number of distinct possible sketches is strictly less than $m^k$. Therefore, at least two strings, say $x(z_1)$ and $x(z_2)$, have the same encoding, for some $z_1,z_2 \in \{0,\ldots,m-1\}^k$ with $z_1 \neq z_2$. Let $i$ be such that $z_1[i] \neq z_2[i]$, and without loss of generality $z_1[i] < z_2[i]$. Now set Bob's input to $y = \pat(i,z_2[i])$, which is a valid subsequence of $x(z_2)$, but not of $x(z_1)$. However, since the encoding for both $x(z_2)$ and $x(z_1)$ is the same, Bob's output will be incorrect for at least one of the strings.
Finally, note that $\lvert y \rvert \leq m \sigma = L$. Hence, we obtain a sketch size lower bound of $\Omega(k \log m) = \Omega(L^{\s-1} \log L)$.

If we instead choose $z$ from $\{0,1\}^k$, then the constructed string $x(z)$ has length $\Oh(k) = \Oh(L^{\s-1})$, and the same argument as above yields a sketch lower bound of $\Omega(L^{\s-1})$.
\end{proof}


We now discuss the complexity of randomized subsequence sketching where Bob is allowed to err with probability $1/3$. To this end, we will reduce from the \emph{Index} problem.

\begin{definition}
 In the \emph{Index} problem, Alice is given an $n$-bit string $z \in \left\{0,1\right\}^n$ and sends a message to Bob. Bob is given Alices's message and an integer $i \in [n]$ and outputs $z[i]$.
\end{definition}

Intuitively, since the communication is one-sided, Alice cannot infer $i$ and therefore has to send the whole string $z$. This intuition also holds for randomized protocols, as follows.

\begin{fact}[\cite{kremer99index}]
 \label{hardness for index}
 The randomized one-way communication complexity of {Index} is $\Omega(n)$.
\end{fact}

Claim \ref{pattern claim} shows that {subsequence sketching} allows us to infer the bits of an arbitrary string~$z$, and thus the hardness of Index carries over to subsequence sketching.
  
\begin{theorem}
\label{One sided randomized communication}
 In a \emph{randomized} subsequence sketch, Bob is allowed to err with probability $1/3$.
 Any randomized subsequence sketch has size $\Omega(L^{\s -1})$ in the worst case.
 This holds even restricted to strings of length $\Theta(L^{\s-1})$.
\end{theorem}

\begin{proof}
We reduce the {Index} problem to {subsequence sketching}. Let $z \in \left\{0,1\right\}^k$ be the input to Alice in the {Index} problem, where $k = m^{\s-1}$. As above, we construct the corresponding input $x(z)$ to Alice in {subsequence sketching}. Observe that $\lvert x(z) \rvert = \mathcal{O}(m^{\s-1})$. For any input $i$ to Bob in the {Index} problem, we construct the corresponding input $\pat(i,0)$ for Bob in {subsequence sketching}. We have $\pat(i,0) \preceq x(z)$ if and only if $z[i] = 1$ (by Claim \ref{pattern claim}). This yields a lower bound of $\Omega(k) = \Omega(m^{\s-1}) = \Omega(L^{\s-1})$ on the sketch size (by Fact \ref{hardness for index}).
\end{proof} 

\section{Weighted LCS}

\begin{definition}
\label{WLCS}
In the WLCS problem we are given strings $x,y$ of lengths $n,m$ over alphabet $\Sigma$ and given a function $W \colon \Sigma \rightarrow \mathbb{N}$. A weighted longest common subsequence (WLCS) of $x$ and $y$ is any string $z$ with $z \preceq x$ and $z \preceq y$ maximizing $W(z) = \sum_{i=1}^{\lvert z \rvert } W(z[i])$. The task is to compute this maximum weight, which we abbreviate as $\WLCS(x,y)$.
\end{definition}
In the remainder of this section we will design an algorithm for computing $\WLCS(x,y)$ in time $\mathcal{O}(\min\{nm,n+m^{\s}\})$. This yields the upper bound of Theorem~\ref{thm:main2}. Note that here we focus on computing the maximum weight $\WLCS(x,y)$; standard methods can be applied to reconstruct a subsequence attaining this value.
We prove a matching conditional lower bound of $\min\{nm, n+m^{\s}\}^{1-o(1)}$ in the next section.

Let $x,y,W$ be given.
The standard dynamic programming algorithm for determining $\LCS(x,y)$ in time $\mathcal{O}(nm)$ trivially generalizes to $\WLCS(x,y)$ as well. Alternatively, we can first compress $x$ to $x' := C_m(x)$ in time $\mathcal{O}(n)$ and then compute the $\WLCS(x',y)$, which is equal to $\WLCS(x,y)$ since all subsequences of length at most $m$ of $x$ are also subsequences of $C_m(x)$. We show below in Theorem~\ref{WLCS of compressed and uncompressed string} how to compute WLCS of a run-length encoded string $x'$ with $r$ runs and a string $y$ of length $m$ in time $\mathcal{O}(r m)$. Since $x' = C_m(x)$ consists of $O(m^{\s-1})$ runs and the length of $y$ is $m$, we can compute $\WLCS(x,y) = \WLCS(C_m(x),y)$ in time $\mathcal{O}(m^{\s})$. In total, we obtain time $\mathcal{O}(\min\{nm,n+m^{\s}\})$.

It remains to solve WLCS on a run-length encoded string $x$ with $r$ runs and a string $y$ of length $m$ in time $\Oh(rm)$.
For (unweighted) LCS a dynamic programming algorithm with this running time was presented by Liu et al.~\cite{liu2008finding}. We first give a brief intuitive explanation as to why their algorithm does not generalize to WLCS. Let $x = c_1^{\ell_1}c_2^{\ell_2} \ldots c_r^{\ell_r}$ be the run-length encoded string, where $c_i \in \Sigma$, and let $L_i = \sum_{j=1}^{i} \ell_j$. Let $D(i,j) := \WLCS(x[{1 \ldots L_i}],y[{1 \ldots j}])$. Liu et al.'s algorithm relies on a recurrence for $D(i,j)$ in terms of $D(i,j-1)$. Consider an input like $x = ba_1a_2 \cdots a_kb$ and $y=a_1a_2 \cdots a_kbb$ with $W(b) > \sum_{\ell \in [k]} W(a_{\ell})$. Note that $D(k+2,k+1) = \sum_{\ell \in [k]}W(a_{\ell}) + W(b)$, but $D(k+2,k+2) = 2W(b)$. Thus $D(k+2,k+2) = D(k+2,k+1) - \sum_{\ell \in [k]} W(a_{\ell}) + W(b)$.  Therefore, in the weighted setting $D(i,j)$ and $D(i,j-1)$ can differ by complicated terms that seem hard to figure out locally. 
Our algorithm that we develop below instead relies on a recurrence for $D(i,j)$ in terms of $D(i-1,j')$.     

\begin{theorem}
\label{WLCS of compressed and uncompressed string}
Given a run-length encoded string $x$ consisting of $r$ runs, a string $y$ of length~$m$, and a weight function $W \colon \Sigma \rightarrow \mathbb{N}$ we can determine $\WLCS(x,y)$ in time $\mathcal{O}(rm)$.
\end{theorem}

\begin{proof}
We write the run-length encoded string  $x$ as $c_1^{\ell_1}c_2^{\ell_2} \ldots c_r^{\ell_r}$ with $c_i \in \Sigma$ and $\ell_i \ge 1$. Let $L_i = \sum_{j = 1}^i \ell_j$. We will build a dynamic programming table $D$ where $D(i,j)$ stores the value  $\WLCS(x[{1 \ldots L_i}],y[{1 \ldots j}])$. In particular, $D(0,j) = D(i,0) = 0$ for all $i,j$. We will show how to compute this table in $\Oh(1)$ (amortized) time per entry in the following.
Since we can split $\WLCS(x[{1 \ldots L_i}],y[{1 \ldots j}]) = \max_{0 \le k \le j} \WLCS(x[{1 \ldots L_{i-1}}], y[{1 \ldots k}]) + \WLCS(c_i^{\ell_i}, y[{k+1 \ldots j}])$, we obtain the recurrence
$D(i,j) = \max_{0 \leq k \leq j} D(i-1,k) + W(c_i) \cdot \min\{\ell_i, |y[{k+1 \ldots j}]|_{c_i}\}$. Since $D(i,j)$ is monotonically non-decreasing in $i$ and $j$, we may rewrite the same recurrence as 
\begin{align*}
 D(i,j) &= \max_{0 \le k \leq j \;:\; |y[{k+1 \ldots j}]|_{c_i} \leq \ell_i} D(i-1,k) + W(c_i) \cdot \lvert y[{k+1 \ldots j}]\rvert_{c_i}.\\
 &= W(c_i) \cdot \lvert y[{1 \ldots j}] \rvert_{c_i} + \max_{0 \le k \leq j \;:\; |y[{k+1 \ldots j}]|_{c_i} \leq \ell_i} D(i-1,k)  - W(c_i) \cdot \lvert y[{1 \ldots k}] \rvert_{c_i}
\end{align*} 
 Let $b_{i,j}$ be the minimum value of $0 \le k \le j$ such that $\lvert y[{k+1 \ldots j}] \rvert_{c_i} \leq \ell_i$. Note that $b_{i,j}$ is well-defined, since for $k = j$ we always have $\lvert y[{k+1 \ldots j}] \rvert_{c_i} = 0 \leq \ell_i$, and note that $b_{i,j}$ is monotonically non-decreasing in $j$. 
 We define the \emph{active $k$-window} $K_{i,j}$ as the interval $\{b_{i,j}, b_{i,j} + 1, \ldots , j\}$. Note that $K_{i,j}$ is non-empty and both its left and right boundary are monotonic in $j$.
 Let $h_i(k) := D(i-1,k) - W(c_i) \cdot \lvert y[{1 \ldots k}] \rvert_{c_i}$ be the \emph{height} of $k$. 
 We define $\text{\emph{highest}}(K_{i,j})$ as $\max_{k \in K_{i,j}} h_i(k)$.
 With this notation, we can rewrite the above recurrence as 
 \[ D(i,j) = W(c_i) \cdot \lvert y[{1 \ldots j}] \rvert_{c_i} + \text{\emph{highest}}(K_{i,j}). \]
 We can precompute all values $\lvert y[{1 \ldots j}] \rvert _c$ in $\mathcal{O}(m)$ time. 
 Hence, in order to determine $D(i,j)$ in amortized time $\mathcal{O}(1)$ it remains to compute $\text{\emph{highest}}(K_{i,j})$ in amortized time $\Oh(1)$. To this end, we maintain the \emph{right to left maximum sequence} of the active window $K_{i,j}$. Specifically, we consider the sequence $\rtlm(K_{i,j}) = \langle k_s,k_{s-1}, \ldots ,k_1 \rangle$ where $k_1 = j$ and  for any $p > 1$, $k_p$ is is the largest number in $K_{i,j}$ with $k_{p} < k_{p-1}$ and $h_i(k_p) > h_i(k_{p-1})$. In particular, $k_s$ is the largest number in $K_{i,j}$ attaining $h_i(k_s) = \text{\emph{highest}}(K_{i,j})$. 
 Hence, from this sequence $\rtlm(K_{i,j})$ we can determine $\text{\emph{highest}}(K_{i,j})$ and thus $D(i,j)$ in time $\Oh(1)$. It remains to argue that we can maintain $\rtlm(K_{i,j})$ in amortized time $\Oh(1)$ per table entry.
 We sketch an algorithm to obtain $\rtlm(K_{i,j})$ from $\rtlm(K_{i,j-1})$.
\begin{algorithm}[H]
\begin{algorithmic}[1]
  \STATE Initialize $\rtlm(K_{i,j}) = \rtlm(K_{i,j-1})$
  \WHILE{the smallest (=leftmost) element $k$ of $\rtlm(K_{i,j})$ satisfies $|y[{k+1 \ldots j}]|_{c_i} > \ell_i$}
            \STATE Remove $k$ from $\rtlm(K_{i,j})$
  \ENDWHILE
  \WHILE{the largest (=rightmost) element $k$ of $\rtlm(K_{i,j})$ satisfies $h_i(k) \le h_i(j)$}
            \STATE Remove $k$ from $\rtlm(K_{i,j})$
  \ENDWHILE
  \STATE Append $j$ to $\rtlm(K_{i,j})$
\end{algorithmic}
\caption{Computing $\rtlm(K_{i,j})$ from $\rtlm(K_{i,j-1})$}
\end{algorithm}

It is easy to see correctness, since the first while loop removes right to left maxima that no longer lie in the active window, the second while loop removes right to left maxima that are dominated by the new element $j$, and the last line adds $j$.
Note that $\lvert y[{k+1 \ldots j}] \rvert_{c} = \lvert y[{1 \ldots j}] \rvert_{c} - \lvert y[{1 \ldots k}] \rvert_{c}$ can be computed in time $\Oh(1)$ from the precomputed values $\lvert y[{1 \ldots j}] \rvert_{c}$, and thus the while conditions can be checked in time $\Oh(1)$. A call of Algorithm 3 can necessitate multiple removal operations, but only one insertion. By charging removals to the insertion of the removed element, we see that Algorithm 3 runs in amortized time $\Oh(1)$. 
We therefore can compute each table entry $D(i,j)$ in amortized time $\Oh(1)$ and obtain total time $\Oh(rm)$.
Pseudocode for the complete algorithm is given below.
\end{proof}

\begin{algorithm}[H]
\begin{algorithmic}[1]
  \STATE \textbf{precompute} $\lvert y[{1 \ldots i}] \rvert_{c}$ for all $i \in [m]$ and $c \in \Sigma$.
  \STATE \textbf{set} $D(i,0) = D(0,j) = 0$ for any $0 \le i \le r$ and $0 \le j \le m$.
  \FOR{$i = 1, \ldots, r$}
    \STATE $\rtlm(K_{i,0}) \gets \langle 0 \rangle$.
    \FOR{$j=1, \ldots, m$} 
       \STATE Update $\rtlm(K_{i,j})$ as in Algorithm 3
       \STATE Let $k$ be the smallest (=leftmost) element of $\rtlm(K_{i,j})$
       \STATE Compute $highest(K_{i,j}) = h_i(k) = D(i-1,k) - W(c_i) \cdot \lvert y[{1 \ldots k}] \rvert_{c_i}$
       \STATE $D(i,j) \gets W(c_i) \cdot \lvert y[{1 \ldots j}] \rvert_{c_i} + highest(K_{i,j})$.
     \ENDFOR
  \ENDFOR
  \STATE \textbf{return} $D(r,m)$
\end{algorithmic}
\caption{Computing $\WLCS(x,y)$ in time $\mathcal{O}(r\cdot m)$}
\end{algorithm}

\section{Conditional lower bound for Weighted LCS}
\label{sec:conlowerbound}

In this section, we prove a conditional lower bound for Weighted LCS, based on the standard hypothesis SETH, which was introduced by Impagliazzo, Paturi, and Zane~\cite{ImpagliazzoPZ01} and asserts that satisfiability has no algorithms that are much faster than exhaustive search.

\medskip
\noindent
\textbf{%
Strong Exponential Time Hypothesis (SETH):} \emph{
For any $\varepsilon > 0$ there is a $k \ge 3$ such that $k$-SAT on $n$ variables cannot be solved in time $\Oh((2-\varepsilon)^n)$.
} 

\medskip
Essentially all known SETH-based lower bounds for polynomial-time problems (e.g.~\cite{AbboudBVW15,backurs2016regular,Bringmann14,BringmannK15,BKsoda18}) use reductions via the \emph{Orthogonal Vectors problem} (OV):
Given sets $A$, $B\subseteq \{0,1\}^D$ of size $|A| = N$, $|B| = M$, determine whether there are $a \in A, b \in B$ that are orthogonal, i.e., $\sum_{i=1}^D a[i] \cdot b[i] = 0$, where the sum is over the integers.
Simple algorithms solve OV in time $\Oh(2^D (N+M))$ and $\Oh(N M D)$. The fastest known algorithm for $D=c(N)\log N$ runs in time $N^{2-1/\Oh(\log c(N))}$ (when $N=M$) \cite{AbboudWY15}, which is only slightly subquadratic for $D \gg \log N$. This has led to the following reasonable hypothesis.

\medskip
\noindent
\textbf{%
(Unbalanced) Orthogonal Vectors Hypothesis (OVH):} \emph{
For any $\gamma > 0$, 
OV restricted to $M = \Theta(N^\gamma)$ and $D=N^{o(1)}$ requires time $(NM)^{1 - o(1)}$.
} 

\medskip
A well-known reduction by Williams~\cite{Williams05} shows that SETH implies OVH in case $\gamma = 1$. Moreover, an observation in~\cite{BKsoda18} shows that if OVH holds for some $\gamma > 0$ then it holds for all $\gamma > 0$.
Thus, OVH is a weaker assumption than SETH, and any OVH-based lower bound also implies a SETH-based lower bound. The conditional lower bound in this section does not only hold assuming SETH, but even assuming the weaker OVH.

We use the following construction from the OVH-based lower bound for LCS~\cite{AbboudBVW15,BringmannK15}. For binary alphabet, such a construction was given in~\cite{BringmannK15}.

\begin{theorem}
\label{LCS-binary}
Given $A,B \subseteq \left\{0,1\right\}^D$ of size $N$, in time $\Oh(DN)$ we can compute strings $x_A$ and $y_B$ on alphabet $\left\{0,1\right\}$ of length $\Theta(DN)$ as well as a number $\tau$ such that $\LCS(x_A,y_B) \geq \tau$ holds if and only if there is an orthogonal pair of vectors in $A$ and $B$. In this construction, $x_A$ and $y_B$ depend only on $A$ and $B$, respectively, and $|x_A|, |y_B|, \tau$ depend only on $N,D$.
\end{theorem}

We now prove a conditional lower bound for WLCS, i.e., the lower bound of Theorem~\ref{thm:main2}.

\begin{theorem} \label{thm:condlowerbound}
Given strings $x,y$ of lengths $n,m$ with $n \ge m$ over alphabet $\Sigma$, computing $\WLCS(x,y)$ requires time $\min\{nm, n + m^\s\}^{1-o(1)}$, assuming OVH. This holds even restricted to $n = m^{\alpha \pm o(1)}$ and $\lvert \Sigma \rvert = \sigma$ for any fixed constants $\alpha \in \mathbb{R}, \alpha \ge 1$ and $\sigma \in \mathbb{N}, \sigma \ge 2$.
\end{theorem}

\begin{proof}
Let $\Sigma = \left\{0,1, \ldots ,\sigma-1 \right\}$ and $\alpha= \alpha_I + \alpha_F$, where $\alpha_I = \lfloor \alpha \rfloor$ and $\alpha_F = \alpha - \alpha_I$ are the integral and fractional parts. Let $M \in \mathbb{N}$ and set $N = \min\{M^{\alpha_I} \cdot \lceil M^{\alpha_F} \rceil, M^{\sigma-1}\}$. Note that $M$ divides $N$. 
Consider any instance $A = \left\{a_0,a_1, \ldots, a_{N-1} \right\} \subseteq \left\{0,1\right\}^d$ and $B = \left\{b_0,b_1, \ldots, b_{M-1} \right\} \subseteq \left\{0,1\right\}^D$ of the Orthogonal Vectors problem. 
Partition $A$ into $A^0,A^1,\ldots,A^{N/M-1}$, where $|A^i| = M$.
Then by Theorem \ref{LCS-binary} we can construct strings $x_A^{(i)}$ and $y_B$ on alphabet $\left\{0,1\right\}$ of length $\Theta(D M)$ and $\tau \in \mathbb{N}$ in time $\mathcal{O}(D M)$  such that $A^i$ and $B$ contain an orthogonal pair of vectors if and only if $\LCS(x_A^{(i)},y_B) \geq \tau$. Note that $A$ and $B$ contain an orthogonal pair of vectors if and only if for some $0 \le i < \frac{N}{M}$, $A^i$ and $B$ contain an orthogonal pair of vectors. Hence, $A$ and $B$ contain an orthogonal pair if and only if $\max_{0 \leq i < \frac{N}{M}} \LCS(x_A^{(i)},y_B) \geq \tau$. 
In the following, we encode the latter inequality into an instance of $\WLCS$.

For simplicity we only give the proof for integral $\alpha$ and $\alpha < \sigma$ (the remaining cases are omitted and can be found in the appendix). In this case, $N = M^\alpha$ and the running time lower bound that we will prove is $(nm)^{1-o(1)}$.

We set $\lambda$ to any value such that $\lambda > |y_B| / M$, and note that $\lambda \in \Theta(D)$ suffices.
Set $W(k) = \lambda \cdot M^{k-1}$ for $k \geq 2$, and $W(1)=W(0)=1$.  Let $\Sigma_k = \left\{0,1, \ldots ,k-1 \right\}$. 
We construct strings $x$ and $y$  as follows:
\begin{align*} 
 x &= x^{(\alpha,0)}\\
 x^{(k,i)} &= \Big( \bigcirc_{j=0}^{M-2} x^{(k-1,M \cdot i + j)} \circ k\Big) \circ x^{(k-1,M \cdot i + (M-1))} \quad  \text{for }2 \leq k\leq \alpha\\ 
 x^{(1,i)} &= x_A^{i} \quad \text{for }0 \leq i < N/M \\
%
 y &= y^{(\alpha)}\\
 y^{(k)} &= k^{M-1} \circ y^{(k-1)} \circ k^{M-1} \quad \text{for }2 \leq k\leq \alpha\\ 
 y^{(1)} &= y_B.
\end{align*} 

Observe that for all $k$,  $x^{(k,i)}$ and $y^{(k)}$ are defined on $\Sigma_k$. In particular, since $\alpha \leq \sigma-1$ we only use symbols from $\Sigma$. Let $\ell(k)$ denote the length of $x^{(k,i)}$ for any $i$. Observe that $\ell(k) = M \cdot \ell(k-1) + (M-1)$ and $\ell(1) \in \Theta(D M)$. Thus, $\ell(k) \in \Theta(D M^{k})$ and $n := \lvert x \rvert \in \Theta(D M^{\alpha})$. It is straightforward to see that $m := \lvert y \rvert = \Theta((k + D)M) = \Theta(D M)$, since $k \le \s = \Oh(1)$. 
Recall that for any string $z$, $W(z)$ is its total weight.

\begin{claim}
\label{technical claim}
For any integer $2 \le k \le \alpha$, we have (1) $(M-1) \cdot \sum_{\ell=2}^{k} W(\ell) = \lambda (M^k-M)$  and (2) $W(y^{(k)}) < W(k+1) + \lambda \cdot (M^k-M)$.
\end{claim}

\begin{proof}
For (1), we calculate
$(M-1) \cdot \sum_{\ell=2}^{k} W(\ell) = (M-1) \cdot \sum_{\ell=1}^{k-1} \lambda M^{\ell}
=\lambda(M^k -M)$. 
For (2), by definition of $y^{(k)}$ and $\lambda$ we have
\[ W(y^{(k)}) < \lambda M + 2(M-1) \cdot \sum \limits_{\ell=2}^{k-1} W(l)
\stackrel{(1)}{=} \lambda M + 2 \lambda (M^k-M) = W(k+1) + \lambda (M^k-M). \qedhere \]
\end{proof}

We now can perform the core step of our correctness argument.

\begin{lemma}
\label{technical}
For any $2 \le k \le \alpha$ and $0 \le i < M^{k-1}$, we have (1) $\WLCS(x^{(k,i)},y^{(k)}) \geq \lambda (M^{k} -M)$, and (2) $\WLCS(x^{(k,i)},y^{(k)}) = (M-1) \cdot W(k) + \WLCS(x^{(k-1,j)},y^{(k-1)})$ for some $M\cdot i \leq j < M\cdot (i+1)$.
\end{lemma}

\begin{proof}
For (1), clearly $\bigcirc_{j=2}^{k}j^{M-1}$ is a common subsequence of $x^{(k,i)}$ and $y^{(k)}$. Together with Claim~\ref{technical claim}.(1), we obtain 
$\WLCS(x^{(k,i)},y^{(k)}) \ge \sum_{j=2}^{k}(M-1) \cdot W(j) = \lambda(M^k-M)$. 

For (2), we claim that $k^{M-1}$ is a subsequence of any WLCS of $x^{(k,i)}$ and $y^{(k)}$. Assuming otherwise, the WLCS can contain at most $M-2$ symbols $k$ and all of $y^{(k-1)}$. Therefore, 
\begin{align*}
 \WLCS(x^{(k,i)},y^{(k)}) &\leq (M-2) \cdot W(k) + W(y^{(k-1)})\\
 &< (M-2) \cdot W(k) + W(k) + \lambda(M^{k-1}-M) \quad \text{by Claim~\ref{technical claim}.(2)}\\
 &= (M-1) \cdot \lambda M^{k-1} + \lambda(M^{k-1}-M)
 = \lambda\cdot (M^{k}-M).
\end{align*} 
This contradicts $\WLCS(x^{(k,i)},y^{(k)}) \ge \lambda(M^k-M)$. It follows that $k^{M-1}$ is a subsequence of the WLCS of $x^{(k,i)}$ and $y^{(k)}$. Hence, $\WLCS(x^{(k,i)},y^{(k)}) = (M-1) \cdot W(k) + \WLCS(x^{(k-1,j)},y^{(k-1)})$ for some $j$ with $M\cdot i \leq j < M \cdot (i + 1)$. 
\end{proof}

Recursively applying the above lemma and substituting $x^{(1,j)}$ by $x_A^{j}$, we conclude that $\WLCS(x,y) = \lambda \cdot (M^{\alpha} -M) + \max_{0 \leq j < M^{\alpha-1}} \LCS(x_A^j,y_B)$. 
Using $M^\alpha = N$ and the construction of $x_A^j, y_B$, we obtain that $\WLCS(x,y) \geq \lambda (N - M) + \tau$ holds if and only if there is an orthogonal pair of vectors in $A$ and $B$. 
Since OVH asserts that solving the OV instance $(A,B)$ in the worst case requires time $(NM)^{1-o(1)}$, even for $D = N^{o(1)}$, we obtain that determining $\WLCS(x,y)$ requires time 
$(NM)^{1-o(1)} = (n m / D^2)^{1-o(1)} = (n m)^{1-o(1)}$. This completes the proof for all instances where $\alpha < \sigma$ is integral.
Note that if $\alpha \geq \sigma$, the claimed lower bound trivially holds as it matches the input size. Now we consider the two remaining cases, where $\sigma - 1 < \alpha < \sigma$ and $\alpha < \sigma-1$.

\medskip
\noindent
\emph{Case $\sigma - 1 < \alpha < \sigma$:} Then $N = M^{\alpha_I}=M^{\sigma-1}$.  We construct strings $x$ and $y$ as follows:
\begin{align*}
&x = x^{(\alpha_I,0)} \circ \alpha_I \circ 0^{D M^{\alpha}}\\
&y = y^{(\alpha_I)} \circ \alpha_I.
\end{align*}
Again, since $\alpha_I \leq \sigma-1$ the strings $x$ and $y$ only use symbols in $\Sigma$. We now have $n := \lvert x \rvert \in \Theta(D M^{\alpha})$ and $m := \lvert y \rvert \in \Theta(D M)$. Clearly, $\WLCS(x,y) \geq \WLCS(x^{(\alpha_I,0)},y^{(\alpha_I)}) + W(\alpha_I) \geq  W(\alpha_I) +  \lambda (M^{\alpha_I} -M)$. Similar to the proof for integral $\alpha$, we claim that $\alpha_I^{M}$ is a subsequence of the WLCS of $x$ and $y$. Assuming otherwise, the WLCS of $x$ and $y$ contains at most $M-1$ symbols $\alpha_I$ and all of $y^{(\alpha_I -1)}$. Therefore, 
\begin{align*}
\WLCS(x,y) &\leq (M-1) \cdot W(\alpha_I) + W(y^{(\alpha_I-1)})\\
&<(M-1) \cdot W(\alpha_I) + W(\alpha_I) + \lambda(M^{\alpha_I-1}-M) \quad \text{by Claim~\ref{technical claim}.(2)}\\  
&=W(\alpha_I) + \lambda \cdot (M^{\alpha_I}-M^{\alpha_I-1} + M^{\alpha_I-1} -M)
=W(\alpha_I) + \lambda(M^{\alpha_I} - M).
\end{align*}
This contradicts $\WLCS(x,y) \geq W(\alpha_I) +  \lambda (M^{\alpha_I} -M)$. Hence, $\alpha_I^{M}$ is a subsequence of the WLCS of $x$ and $y$, and $\WLCS(x,y) = W(\alpha_I) + \WLCS(x^{(\alpha_I,0)},y^{(\alpha_I)})$. It follows that $\WLCS(x,y) \geq \lambda M^{\alpha_I-1} + \lambda (M^{\alpha_I}-M) + \tau$ holds if and only if there exists an orthogonal pair of vectors in $A$ and $B$. 
OVH asserts that solving the OV instance $(A,B)$ in the worst case requires time $(NM)^{1-o(1)}$, even for $D = N^{o(1)}$. Using $N = \Theta(M^{\alpha_I})= \Theta(M^{\sigma-1})$, we obtain that determining $\WLCS(x,y)$ requires time 
$(NM)^{1-o(1)} = (M^{\sigma})^{1-o(1)} = ((m/D)^\sigma)^{1-o(1)} = (m^\s)^{1-o(1)}$.
This completes the proof in the case $\sigma - 1 < \alpha < \sigma$.

\medskip
\noindent
\emph{Case $\alpha < \sigma-1$:} 
In this case $\alpha_I \leq \sigma -2$ and $N = M^{\alpha_I} \cdot \lceil M^{\alpha_F} \rceil$. Let $f = \lceil M^{\alpha_F} \rceil$ as shorthand. We construct $x$ and $y$ as follows:
\begin{align*}
&x = \Big( \bigcirc_{j=0}^{f-2}x^{(\alpha_I,j)} \circ (\alpha_I+1) \Big) \circ x^{(\alpha_I,f-1)}\\
&y = (\alpha_I+1)^{f} \circ y^{(\alpha_I)} \circ (\alpha_I+1)^{f}
\end{align*}
Once again $x$ and $y$ consist of symbols in $\Sigma$, since $\alpha_I \leq \sigma-2$. Since $|x^{(\alpha_I,i)}| \in \Theta(D M^{\alpha_I})$, we have $n := \lvert x \rvert \in \Theta(D M^{\alpha_I + \alpha_F}) = \Theta(D M^{\alpha})$, and $m := \lvert y \rvert \in \Theta(D M)$. 
The same argument as before, now with $f$ instead of $M$ parts, shows that 
%
$\WLCS(x,y) = 
(f-1)W(\alpha_I+1) + \WLCS(x^{(\alpha_I,j)}, y^{(\alpha_I)})$ holds for some $0 \le j < f$.
Plugging in $\WLCS(x^{(\alpha_I,j)}, y^{(\alpha_I)})$, we see that 
\[ \WLCS(x,y) = \lambda(f-1)M^{\alpha_I} + \lambda(M^{\alpha_I} -M) + \max_{0 \leq j \leq \frac{N}{M} -1} \LCS(x_A^{j},y_B).\] 
Hence, $\WLCS(x,y) \ge \lambda (f-1)M^{\alpha_I} + \lambda(M^{\alpha_I} -M) + \tau$ holds if and only if there is an orthogonal pair of vectors in $A$ and $B$. 
OVH asserts that solving the OV instance $(A,B)$ in the worst case requires time $(NM)^{1-o(1)}$, even for $D = N^{o(1)}$. Using $N = \Theta(M^{\alpha_I} \cdot f) = \Theta(M^{\alpha})$, we obtain that determining $\WLCS(x,y)$ requires time 
$(NM)^{1-o(1)} = (M^{\alpha + 1})^{1-o(1)} = (n m/D^2)^{1-o(1)} = (nm)^{1-o(1)}$.
This completes the proof of the last case $\alpha < \sigma-1$.

Finally, note that in all cases we constructed strings over alphabet size $\sigma$ of length $n = M^{\alpha\pm o(1)}$ and $m = M^{1 \pm o(1)}$, and thus $n = m^{\alpha \pm o(1)}$.
\end{proof}

%

%
%



\end{document}